\documentclass[11pt]{article}

\usepackage{color,amsthm,amsmath,amsfonts,caption,tikz}

\usepackage{hyperref}
\usepackage{fullpage}
\usepackage[toc,page]{appendix}
\usepackage{authblk}

\title{
Fundamental Limits of Universal Variable-to-Fixed Length Coding of Parametric Sources \thanks{This research was funded in part by the
  NSF under grant No. CCF-1422358.}}
\author[1]{Nematollah Iri\thanks{ niri1@asu.edu}}
\author[2]{Oliver Kosut\thanks{okosut@asu.edu}}

\affil[1,2]{Arizona State University}

\newtheorem{Theorem}{Theorem}
\newtheorem{Corollary}{Corollary}

\newtheorem{Example}{Example}
\newtheorem{Lemma}{Lemma}
\newtheorem{Remark}{Remark}

\begin{document}

\maketitle

\begin{abstract}
Universal variable-to-fixed (V-F) length coding of $d$-dimensional exponential family of distributions is considered. We propose an achievable scheme consisting of a dictionary, used to parse the source output stream, making use of the previously-introduced notion of \emph{quantized types}. The quantized type class of a sequence is based on partitioning the space of minimal sufficient statistics into cuboids. Our proposed dictionary consists of sequences in the boundaries of transition from low to high quantized type class size. We derive the asymptotics of the $\epsilon$-coding rate of our coding scheme for large enough dictionaries. In particular, we show that the third-order coding rate of our scheme is $H\frac{d}{2}\frac{\log\log M}{\log M}$, where $H$ is the entropy of the source and $M$ is the dictionary size. We further provide a converse, showing that this rate is optimal up to the third-order term.
\end{abstract}
\section{Introduction}
A variable-to-fixed (V-F) length code consists of a dictionary of pre-specified size. Elements of the dictionary (segments) are used to parse the infinite sequence emitted from the source. Segments may have variable length, however they are encoded to the fixed-length binary representation of their indices within the dictionary. In order to be able to uniquely parse any infinite length sequence into the segments, we assume the dictionary to be complete (i.e. every infinite length sequence has a prefix within the dictionary) and proper (i.e. no segment is a prefix of another segment). The underlying source model induces a distribution on the segment lengths. The segment length distribution reflects the quality of the dictionary for the compression task.

For a given memoryless source, Tunstall \cite{tunstall} provided an average-case optimal algorithm to maximize average segment length. A central limit theorem for the Tunstall algorithm's code length has been derived in \cite{drmota}. In most applications, however, statistics of the source are unknown or arduous to estimate, especially at short blocklengths, where there are limited samples for the inference task. In universal source coding, the underlying distribution in force is unknown, yet belongs to a known collection of distributions. Universal V-F length codes are studied in e.g. \cite{krichevsky, Lawrence, tjalkens, viswer}. Upper and lower bounds on the \emph{redundancy} of a universal code for the class of \emph{all} memoryless sources is derived in \cite{krichevsky}. Universal V-F length coding of the class of all binary memoryless sources is then considered in \cite{Lawrence, tjalkens}, where \cite{tjalkens} provides an asymptotically average sense optimal\footnote{Throughout, ``optimality'' of an algorithm is considered only up to the model cost term (i.e. the term reflecting the price of universality) in the coding rate. The model cost term is the second-order term in the average case analysis, while it is the third-order term in the probabilistic analysis.} algorithm. Later, optimal redundancy for V-F length compression of the class of Markov sources is derived in \cite{viswer}. Performance of V-F length codes and fixed-to-variable (F-V) length codes for compression of the class of Markov sources is compared in \cite{merhavVFvsFV} and a dictionary construction that asymptotically achieves the optimal error exponent is proposed.

All previous works consider model classes that include all distributions within a simplex. However, universal V-F length coding for \emph{more structured} model classes has not been considered in the literature. Apart from extending the topological complexities, we further adopt more general metrics of performance. Delay-sensitive modern applications reflect new requirements on the performance of compression schemes. Therefore it is vital to characterize the overhead associated with operation in the non-asymptotic regime. Over the course of probing the non-asymptotics, incurring ``errors'' are inevitable. Therefore, we depart from classical average-case (redundancy) and worst case (regret) analysis to the modern probabilistic analysis, where the figure of merit in our setup is the $\epsilon$-coding rate --- the minimum rate such that the corresponding overflow probability is less than $\epsilon$. Our goal is to analyze asymptotics of the $\epsilon$-coding rate as the size of the dictionary increases. We provide an achievable scheme for compressing $d$-dimensional exponential family of distributions as the parametric model class. Moreover, we provide a converse result, showing that our proposed scheme is optimal up to the third-order $\epsilon$-coding rate.

In previous universal V-F length codes, one can define a notion of complexity for sequences. In \cite{krichevsky, Lawrence, tjalkens, viswer}, a sequence with high complexity has low probability under a certain composite or mixture source. While in \cite{merhavVFvsFV}, high complexity sequences have high scaled (by sequence length) empirical entropy. The dictionary of such algorithms then consists of sequences in the boundaries of transition from low complexity to high complexity. We follow a similar complexity theme to design the dictionary. The sequence complexity in our proposed algorithm is characterized based on the sequence's type class size, hence we name our scheme the \emph{Type Complexity} (TC) code. Scaled empirical entropy \cite{merhavVFvsFV} is ignorant of the underlying structure of the parametric class. Therefore, in order to \emph{fully} exploit the inherited structure of the model class, we characterize type classes based on quantized types, which we introduced in \cite{nematArxiv,nemat2} in studying F-V length compression. We partition the space of minimal sufficient statistics into cuboids, and define two sequences to be in the same quantized type class if and only if their minimal sufficient statistic falls within the same cuboid.

The type class approach has been taken before for the compression problem in \cite{oliver}. The Type Size code (TS code) is introduced in \cite{oliver} for F-V length compression of the class of all stationary memoryless sources, in which sequences are encoded in increasing order of type class sizes. The exquisite aspect of this approach is the freedom in defining types. In fact, for F-V length coding, any universal one-to-one compression algorithm can be considered as a TS code with a proper characterization of types \cite{niriCiss18}. In \cite{nematArxiv}, we considered universal F-V length source coding of parametric sources. We have shown \cite{nematArxiv} that the TS code using quantized types achieves optimal coding rate for F-V length compression of the exponential family of distributions.

In this work, we provide a performance guarantee for V-F length compression of the exponential family using our proposed Type Complexity code. We upper bound the $\epsilon$-coding rate of the quantized type implementation of the Type Complexity code by
\begin{equation}
\label{thisEquation}
H+\sigma\sqrt{\frac{H}{\log{M}}}Q^{-1}(\epsilon)+H\frac{d}{2}\frac{\log\log{M}}{\log{M}}+\mathcal{O}\left(\frac{1}{\log{M}}\right)
\end{equation}
where $H,\sigma^2$ are the entropy and the varentropy of the underlying source, respectively, $M$ is the pre-specified dictionary size, $Q(\cdot)$ is the tail of the standard normal distribution, and $d$ is the dimension of the model class. We then provide a converse result showing that this rate is optimal up to the third-order term. Our converse proof relies on the construction of a F-V length code from a V-F length code presented in \cite{merhavVFvsFV}, along with a converse result for F-V length prefix codes \cite{kosutJournal}.

Comparing the third-order term in (\ref{thisEquation}) with Rissanen's \cite{rissanenUIPE} redundancy $\frac{d}{2}\frac{\log{n}}{n}$ for F-V length codes, where $2^n$ denotes the fixed number of codewords in the F-V length code and plays the role of $M$ (fixed number of segments in the V-F length code), we observe that for binary memoryless sources, the optimal V-F length code provides better convergence for the model cost term than the F-V length codes, while for sources with $H>1$, the optimal F-V length code trumps the V-F length codes from the perspective of model cost effects. On the other hand, comparing the dispersion term in (\ref{thisEquation}) with the dispersion of the optimal F-V length code \cite{nematArxiv}, which is $\frac{\sigma}{\sqrt{n}}Q^{-1}(\epsilon)$, we observe that the optimal V-F length code provides better dispersion for binary memoryless sources, while for sources with $H>1$, optimal F-V length code provides better dispersion effects.

The rest of the paper is organized as follows: In Sec. \ref{sec::ProbState}, we introduce the exponential family, V-F length coding and related definitions. In Sec. \ref{sec::QUanTyp}, we reproduce the characterization of quantized types from \cite{nematArxiv}. Type Complexity code is presented in Sec. \ref{sec::Algorithm}. Main result of the paper is stated in Sec. \ref{sec::MainRwes}. We present preliminary results in Sec. \ref{sec::Prelim}. The Achievability and the converse results are proved in Sec.'s \ref{sec::Achiev} and \ref{sec::Converse}, respectively. We conclude in Sec. \ref{sec::cncld}.

\section{Problem Statement}
\label{sec::ProbState}
Let $\Theta$ be a compact subset of $\mathbb{R}^d$. Probability distributions in an exponential family can be expressed in the form
\begin{equation}
\label{pTheta}
p_{\theta}(x)=2^{\left\langle\theta,\boldsymbol{\tau}(x)\right\rangle - \psi(\theta)}
\end{equation}
where $\theta\in\Theta$ is the $d$-dimensional parameter vector, $\boldsymbol{\tau}(x): \mathcal{X}\rightarrow \mathbb{R}^d$ is the vector of sufficient statistics and $\psi(\theta)$ is the normalizing factor. Let the model class $\mathcal{P}=\left\{p_{\theta},\theta\in\Theta\right\}$, be the exponential family of distributions over the finite alphabet $\mathcal{X}=\left\{1,\cdots,|\mathcal{X}|\right\}$, parameterized by $\theta\in\Theta\subset \mathbb{R}^d$, where $d$ is the degrees of freedom in the minimal description of $p_{\theta}\in\mathcal{P}$, in the sense that no smaller dimensional family can capture the same model class. The degrees of freedom turns out to characterize the richness of the model class in our context. Compactness of $\Theta$ implies existence of uniform bounds $0<p_{\min},p_{\max}<1$ on the probabilities, i.e.
\begin{equation}
p_{\min}\leq p_{\theta}(x)\leq p_{\max} \hspace{0.25in}\forall \theta\in\Theta ,\forall x\in\mathcal{X}.
\end{equation}
Let $X^\infty$ be the infinite length sequence drawn $i.i.d.$ from the (unknown) true model $p_{\theta^*}$. From (\ref{pTheta}), the probability of a sequence $x^{\ell}=x_1\cdots x_{\ell}$ drawn $i.i.d.$ from a model $p_{\theta}\in\mathcal{P}$ in the exponential family takes the form \cite{merhav}
\begin{align}
p_{\theta}(x^{\ell})&=\prod_{i=1}^{{\ell}}{p_{\theta}(x_i)} \nonumber \\
               &= \prod_{i=1}^{{\ell}}{2^{\langle\theta,\boldsymbol{\tau}(x_i)\rangle-\psi(\theta)}} \nonumber \\
               &=2^{\ell\left[\langle\theta,\boldsymbol{\tau}(x^{\ell})\rangle-\psi(\theta)\right]}
\end{align}
where
\begin{equation}
\boldsymbol{\tau}(x^{\ell})=\frac{\sum_{i=1}^{{\ell}}\boldsymbol{\tau}(x_i)}{{\ell}}\in\mathbb{R}^{d}
\end{equation}
is a minimal sufficient statistic \cite{merhav}. Note that $\boldsymbol{\tau}(x)$ and $\boldsymbol{\tau}(x^{\ell})$ are distinguished based upon their arguments. We denote $\mathbb{P}_{\theta}$, $\mathbb{E}_{\theta}$ and $\mathbb{V}_{\theta}$ as the probability, expectation and variance with respect to $p_{\theta}$, respectively. We denote the set of all finite length sequences over $\mathcal{X}$ as $\mathcal{X}^{*}$. We denote the generic source sequence of unspecified length as $x^*\in\mathcal{X}^*$. Let $x^{\ell}x^{\ell'}$ be the concatenation of $x^{\ell}$ and $x^{\ell'}$. All logarithms are in base 2. For a set $\mathcal{B}$, $|\mathcal{B}|$ denotes its size. Instead of introducing different indices for every new constant $C_1, C_2, ...$, the same letter $C$ may be used to denote different constants whose precise values are irrelevant.

A V-F length code consists of a parsing dictionary $\mathcal{D}$ of a pre-specified size $|\mathcal{D}|=M$, which is used to parse the source sequence. Elements of the dictionary (segments), which we denote by $\{x_1^*,\cdots,x_M^*\}$, may have different lengths. Once a segment $x^*\in\mathcal{D}$ is identified as a parsed sequence, it is then encoded to its lexicographical index within $\mathcal{D}$ using $\log{M}$ bits. As it does not hurt our analysis, we ignore rounding $\log{M}$ to its closest integer.

We assume $\mathcal{D}$ is complete, i.e. any infinite length sequence over $\mathcal{X}$ has a prefix in $\mathcal{D}$. In addition, we assume $\mathcal{D}$ is proper, i.e. there are no two segments where one is a prefix of the other. Completeness along with properness of $\mathcal{D}$ implies that any long enough sequence has a unique prefix in the dictionary. Every complete and proper dictionary can be represented with a rooted complete $|\mathcal{X}|$-ary tree in which every internal node has $|\mathcal{X}|$ child nodes. Let us label each of the $|\mathcal{X}|$ edges branching out of an internal node with different letters from $\mathcal{X}$. Each node corresponds to the sequence of edge-labels from the root to the node. One can then correspond internal nodes of the tree to the prefixes of the segments, while leaf nodes correspond to the segments.

Let $\mathcal{D}$ be the dictionary of a V-F length code $\phi$. Let $X^*\in\mathcal{D}$ be the \emph{random} first parsed segment of the source output $X^{\infty}$, using the dictionary $\mathcal{D}$. Let $\ell(X^*)$ be the length of $X^*$. We adopt a one-shot setting and denote
\begin{equation}
\label{OneShotELnmgthEq}
\ell^{\phi}(X^{\infty})=\ell(X^{*}).
\end{equation}
We gauge the performance of V-F length code $\phi$ with a dictionary $\mathcal{D}$ of size $M$, through the $\epsilon$-coding rate given by
\begin{equation}
R_M(\epsilon,\phi,p_{\theta^*}):=\min\left\{R: \:\:\mathbb{P}_{\theta^*}\left(\frac{\log{M}}{\ell^{\phi}(X^{\infty})}\geq R \right)\leq \epsilon\right\}. \label{epsCodRate}
\end{equation}
Our goal is to analyze the behavior of $R_M(\epsilon,\phi,p_{\theta^*})$ for large enough dictionary size $M$.
\begin{Remark}
Optimizing the $\epsilon$-coding rate provides more refined results than optimizing $\frac{\log{M}}{\mathbb{E}_{\theta^*}\left({\ell^{\phi}(X^{\infty})}\right)}$. The latter is done in e.g. \cite{tjalkens}.
\end{Remark}
\section{Quantized Types}
\label{sec::QUanTyp}
We have previously introduced \emph{quantized types} \cite{nematArxiv,nemat2}, the optimal\footnote{Optimality is in the sense that the quantized type class implementation of the TS code achieves the minimum third-order coding rate.} characterization of type classes for the universal F-V length compression of the exponential family. In this section, we briefly review this characterization. In order to define the quantized type class of a sequence $x^{\ell}\in\mathcal{X}^{\ell}$, we cover the convex hull of the set of minimal sufficient statistics $\mathcal{T}=\text{conv}\left\{\boldsymbol{\tau}(x): x\in\mathcal{X}\right\}$, into $d$-dimensional cubic grids --- cuboids --- of side length $\frac{W}{\ell}$, where $W>0$ is a constant. The union of such disjoint cuboids should cover $\mathcal{T}$. The position of these cuboids is arbitrary, however once we cover the space, the covering is fixed throughout. We represent each $d$-dimensional cuboid by its geometrical \emph{center}. Denote $G(\boldsymbol{\tau}_0)$ as the cuboid with center $\boldsymbol{\tau}_0$. More precisely
\begin{equation}
\label{cuboidEq}
G(\boldsymbol{\tau}_0):= \left\{\boldsymbol{z}+\boldsymbol{\tau}_0 \in \mathbb{R}^d: -\frac{W}{2{\ell}}<z_i\leq \frac{W}{2{\ell}} \mbox{ for } 1\leq i \leq d \right\}
\end{equation}
where  $z_i$ is the $i$-th component of the $d$-dimensional vector $\boldsymbol{z}$.
Let $\boldsymbol{\tau}_c(x^{\ell})$ be the center of the cuboid that contains $\boldsymbol{\tau}(x^{\ell})$.

We then define the quantized type class of $x^{\ell}$ as
\begin{equation}
\label{typeClassDefEq}
T_{x^{\ell}}:=\left\{y^{\ell}\in\mathcal{X}^{\ell}: \boldsymbol{\tau}(y^{\ell})\in G\left(\boldsymbol{\tau}_c(x^{\ell})\right)\right\}
\end{equation}
the set of all sequences $y^{\ell}$ with minimal sufficient statistic belonging to the very same cuboid containing the minimal sufficient statistic of $x^{\ell}$ (See Figure \ref{fig::TClassPar}). We denote $\mathcal{T}_{\ell}=\left\{T_{x^{\ell}}:x^{\ell}\in\mathcal{X}^{\ell}\right\}$ as the set of all quantized type classes for sequences of length $\ell$.
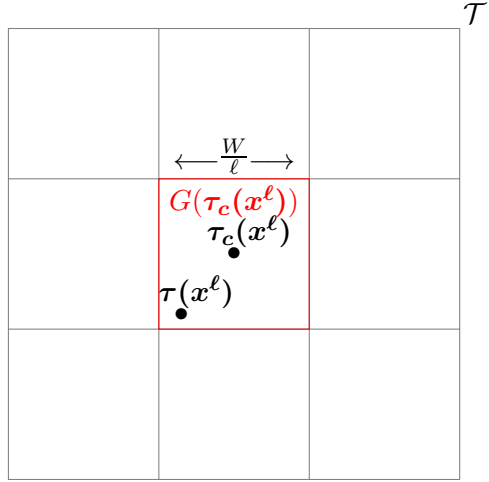
\begin{figure}
\centering
\captionsetup{justification=centering}
\begin{tikzpicture}
\draw[step=2cm,color=gray] (0,0) grid (6,6);
\node at (3,4.3) {$\frac{W}{\ell}$};
\node at (2.5,4.2) {$\longleftarrow$};
\node at (3.5,4.2) {$\longrightarrow$};
\node at (6.2,6.2) {$\mathcal{T}$};
\node at (2.3,2.2) {$\bullet$};
\node at (2.5,2.5) {$\boldsymbol{\tau(x^{\ell})}$};
\node at (3,3) { $\bullet$ } ;
\node at (3.2,3.3) {$\boldsymbol{\tau_c(x^{\ell})}$};
\draw[step=2cm,color=red] (2,2) grid (4,4);
\node at (3,3.7) {{$\color{red}{G(\boldsymbol{\tau_c(x^{\ell})})}$}};
\end{tikzpicture}
\caption{Quantized Types}
\label{fig::TClassPar}
\end{figure}

\section{Type Complexity Code}
\label{sec::Algorithm}
In this section, we propose the Type Complexity (TC) code. Our designed dictionary $\mathcal{D}$, consists of sequences in the boundaries of transition from low quantized type class size to high quantized type class size. More precisely, let $\gamma$ be chosen as the largest positive constant such that the resulting dictionary has at most $M$ segments; we characterize this $\gamma$ precisely in Section \ref{sec::SizeEnf}. The sequence $x^{\ell}=(x_1,x_2,\cdots,x_{\ell})$ is a segment in the dictionary of the TC code if and only if
\begin{equation}
\label{tyCmAlg}
\log{\left|T_{x^{\ell}}\right|}> \gamma \mbox{ and } \log{\left|T_{{x^{\ell}}^{-1}}\right|}\leq\gamma
\end{equation}
where $T_{x^{\ell}}$ is the quantized type class of $x^{\ell}$ as defined in (\ref{typeClassDefEq}) and ${x^{\ell}}^{-1}=(x_1,x_2,\cdots,x_{\ell-1})$ is obtained from $x^{\ell}$ by deleting the last letter.

From construction, it is clear that $\mathcal{D}$ is proper, and furthermore monotonicity of $\log{|T_{x^{\ell}}|}$ in $\ell$ implies completeness of $\mathcal{D}$. Intuitively, sequences with large type class sizes contain more \emph{information}, implying that the TC code compresses more information into a fixed budget of output bits, which is the promise of the optimal V-F length code.

We note that there is a freedom in defining type classes in (\ref{tyCmAlg}). We show that the quantized type is the relevant characterization of type classes for the optimal performance.

\section{Main Result}
\label{sec::MainRwes}
Let $H(p_{\theta})=\mathbb{E}_{\theta}\left(\log{\frac{1}{p_{\theta}(X)}}\right)$ and $\sigma^2(p_{\theta})=\mathbb{V}_{\theta}\left(\log{\frac{1}{p_{\theta}(X)}}\right)$ be the entropy and the varentropy of $p_{\theta}$, repectively. The following theorem exactly characterizes achievable $\epsilon$-rates up to third-order term, as well as asserting that this rate is achievable by the TC code using quantized types.
\begin{Theorem}
\label{mainThm}
For any stationary memoryless exponential family of distributions parameterized by $\Theta$,
\begin{equation}
\inf_{\phi}\sup_{\theta\in\Theta}\left[R_M(\epsilon,\phi,{p_{\theta}})-H(p_{\theta})-\sigma(p_{\theta})\sqrt{\frac{H(p_{\theta})}{\log{M}}}Q^{-1}(\epsilon)-H(p_{\theta})\frac{d}{2}\frac{\log{\log M}}{\log{M}}\right]=o\left(\frac{\log\log M}{\log{M}}\right) \label{mainEq}
\end{equation}
where the infimum is achieved by the TC code using quantized types.
\end{Theorem}

\begin{Example}
For the class of all binary memoryless sources $d=1$, and the third-order term in (\ref{mainEq}) matches with the optimal redundancy in \cite{tjalkens}.
\end{Example}

\section{Preliminary Results}
\label{sec::Prelim}
Define
\begin{equation}
\label{theThaHatEq}
\hat{\theta}\left(\boldsymbol{\tau}\right)	=\underset{\theta\in\Theta}{\arg\max} \left(\langle\theta,\boldsymbol{\tau}\rangle-\psi(\theta)\right).
\end{equation}
Note that since the Hessian matrix of $\psi(\theta)$, $\boldsymbol{\nabla}^2\left(\psi(\theta)\right)=\text{Cov}_{\theta}\left(\boldsymbol{\tau}(X)\right)$ is positive definite, the log-likelihood function is strictly concave and hence the maximum likelihood $\hat{\theta}(\boldsymbol{\tau})$ is unique.

The following lemma, which is a direct consequence of \cite[Lemmas 1 and 3]{nematArxiv} provides tight upper and lower bounds on the quantized type class size.
\begin{Lemma}
\label{TypeClSizeLem}
Size of the quantized type class of $x^{\ell}$ is bounded as
\begin{equation}
				-\log{p_{\hat{\theta}(x^{\ell})}(x^{\ell})}-\frac{d}{2}\log{\ell}+C_1 \leq \log{|T_{x^{\ell}}|} \leq -\log{p_{\hat{\theta}(x^{\ell})}(x^{\ell})}-\frac{d}{2}\log{\ell}+C_2 \label{TypeClassSize}
\end{equation}
where $C_1, C_2$ are constants independent of $\ell$.
\end{Lemma}

The type class size bounds in the previous lemma are springboards to the following upper bound on the lengths of the dictionary segments.
\begin{Corollary}[Segment Length]
There exists a positive constant $C_3>0$, such that for any $x^{\ell}\in\mathcal{D}$, we have
\begin{equation}
\label{LinearEll}
 \ell \leq C_3\gamma.
\end{equation}
\end{Corollary}
\begin{proof}
For any $x^{\ell}\in\mathcal{D}$, (\ref{tyCmAlg},\ref{TypeClassSize}) yield
\begin{equation*}
-\log{p_{\hat{\theta}\left(x^{{\ell}^{-1}}\right)}\left(x^{{\ell}^{-1}}\right)}-\frac{d}{2}\log{(\ell-1)}+C_1\leq\log|{T_{x^{\ell^{-1}}}}|\leq\gamma.
\end{equation*}
Since for all $\theta\in\Theta$, $p_{\theta}(x^{{\ell}^{-1}})\leq p_{\max}^{\ell-1}$, we have
\begin{equation*}
(\ell-1)\log{\frac{1}{p_{\max}}}-\frac{d}{2}\log{(\ell-1)}\leq\gamma-C_1.
\end{equation*}
The corollary then follows.
\end{proof}

The following lemma shows that one single observation does not provide much information.
\begin{Lemma}
\label{dataIncInfo}
Let $x^{\ell+1}=(x_1,\cdots, x_{\ell},x_{\ell+1})=x^{\ell}x_{\ell+1}$. There exists a constant $C_4>0$ such that
\begin{equation}
-\log{p_{\hat{\theta}\left(x^{\ell+1}\right)}\left(x^{\ell+1}\right)}-\left(-\log{p_{\hat{\theta}\left(x^{\ell}\right)}\left(x^{\ell}\right)}\right)\leq C_4.
\end{equation}
\end{Lemma}

\begin{proof}
We have
\begin{align}
&-\log{p_{\hat{\theta}\left(x^{\ell+1}\right)}\left(x^{\ell+1}\right)}-\left(-\log{p_{\hat{\theta}\left(x^{\ell}\right)}\left(x^{\ell}\right)}\right) = \nonumber\\
 &\hspace{1in} \max_{\theta}\left[(\ell+1)\left(\psi(\theta)-\langle\theta,\boldsymbol{\tau}(x^{\ell+1})\rangle\right)\right] -\max_{\theta}\left[\ell\left(\psi(\theta)-\langle\theta,\boldsymbol{\tau}(x^{\ell})\rangle\right)\right] \label{firstLineEqq}\\
 &\hspace{1in}\leq\max_{\theta}\left[(\ell+1)\psi(\theta)-(\ell+1)\langle \theta,\boldsymbol{\tau}(x^{\ell+1})\rangle-\ell\psi(\theta)+\ell\langle\theta,\boldsymbol{\tau}(x^{\ell})\rangle\right] \label{secondLineEqq}\\
 &\hspace{1in} \leq C_4 \label{lastLLLine}
\end{align}
where (\ref{firstLineEqq}) is from the definition (\ref{theThaHatEq}), (\ref{secondLineEqq}) exploits the fact that for any two functions $g_1(\theta),g_2(\theta)$
\begin{equation*}
\max_{\theta}g_1(\theta)-\max_{\theta}g_2(\theta)\leq \max_{\theta}\Big((g_1-g_2)(\theta)\Big),
\end{equation*}
and finally (\ref{lastLLLine}) follows from $|\boldsymbol{\tau}(x^{\ell})-\boldsymbol{\tau}(x^{\ell+1})|\leq \frac{C}{\ell}$ for some constant $C$ along with the fact that $\psi(\theta)$ is a continuous function over a compact domain and hence is bounded.
\end{proof}

We appeal to the following normal approximation result from \cite{verdulossless,saito}, in order to bound the percentiles of the type class size in the achievability proof.
\begin{Lemma}[Asymptotic Normality of Information]\cite{verdulossless,saito}\label{maxLikeBerr}
Fix a positive constant $\alpha>0$. For a stationary memoryless source, there exists a finite positive constant $A>0$, such that for all $\ell\geq 1$ and $z$ with $|z|\leq \alpha$,
\begin{equation}
\left|\mathbb{P}_{\theta^*}\left(\frac{-\log{p_{\theta^*}(X^{\ell})}-\ell H}{\sigma\sqrt{\ell}}>z\right)-Q(z)\right|\leq \frac{A}{\sqrt{\ell}}
\end{equation}
where $H:=H(p_{\theta^*})$ and $\sigma^2:=\sigma^2(p_{\theta^*})$, are the entropy and the varentropy of the true model $p_{\theta^*}$, respectively.
\end{Lemma}

\section{Achievability}
\label{sec::Achiev}
\subsection{Threshold Design}
\label{sec::SizeEnf}
Setting high threshold values of $\gamma$ in (\ref{tyCmAlg}), results in compressing more information into a fixed budget of output bits. On the other hand, in order to keep the dictionary size below the pre-specified size $M$, $\gamma$ cannot be set too high. In this subsection, we characterize the largest value of $\gamma$ for which the resulting dictionary size is below $M$.

Let $N_{\ell+1}$ be the number of dictionary segments with length $\ell+1$. For any $x^{\ell+1}\in\mathcal{D}$, it must certainly hold that $\log{|T_{{x^{\ell+1}}}|}>\gamma$ and $\log{|T_{{x^{\ell+1}}^{-1}}|}\leq\gamma$. Let
\begin{equation}
\mathcal{A}= \{T\in\mathcal{T}_{\ell}: \log{|T|}\leq \gamma \mbox{ and }\exists\: x^{\ell}\in T \mbox{ and }  x_{\ell+1}\in\mathcal{X} \mbox{ with }\log|{T_{x^{\ell}x_{\ell+1}}}|>\gamma\}.
\end{equation}
Motivated by \cite[Eq. 3.12]{merhavVFvsFV}, we upper bound $N_{\ell+1}$ as follows:
\begin{align}
N_{\ell+1} & \leq |\mathcal{X}|\sum_{T\in\mathcal{A}}{|T|} \nonumber\\
            &\leq |\mathcal{X}|2^{\gamma} |\mathcal{A}|. \label{setCombis}
\end{align}
We show in Appendix \ref{app::SizeAi} that $|\mathcal{A}|\leq\ell^{d-1}$. Hence
\begin{equation}
N_{\ell+1}\leq |\mathcal{X}|2^{\gamma}\ell^{d-1}. \label{eNSizeUJp}
\end{equation}
We then upper bound the dictionary size as follows:
\begin{align}
|\mathcal{D}|&= \sum_{\ell=0}^{C_3\gamma} N_{\ell+1}  \label{firstLineEq}\\
&\leq |\mathcal{X}|2^{\gamma}\sum_{\ell=0}^{C_3\gamma}{\ell^{d-1}} \label{secondLine}\\
                                   &\leq C2^{\gamma}\gamma^{d} \label{cnoRel}
\end{align}
where (\ref{firstLineEq}) is from (\ref{LinearEll}), (\ref{secondLine}) follows from (\ref{eNSizeUJp}), and (\ref{cnoRel}) is a consequence of upper bounding the summation with an integral, where $C>0$ is a generic constant whose precise value is irrelevant. Finally, to ensure that the dictionary of the quantized Type Complexity code (\ref{tyCmAlg}) does not contain more than $M$ segments, it suffices to set $\gamma$ such that
\begin{equation}
\label{notBigEmEq}
\log{C}+\gamma+d\log{\gamma}\leq\log{M}.
\end{equation}
One can show that, there exists a positive constant $C>0$, such that the following choice of $\gamma$, satisfies (\ref{notBigEmEq}) and moreover the leading two terms are the largest possible:
\begin{equation}
\label{gammaValued}
\gamma = \log{M}-d\log\log{M}-C.
\end{equation}
\subsection{Coding Rate Analysis}
In this subsection, we derive an upper bound for the $\epsilon$-coding rate of the quantized type implementation of the TC code. To this end, we upper bound the overflow probability as follows:
\begin{align}
\mathbb{P}\left(\frac{\log{M}}{\ell(X^*)}>R\right)&= \mathbb{P}\left(\ell(X^*)<\frac{\log{M}}{R}\right) \nonumber \\
                                          &=\mathbb{P}\left(\exists \ell<\frac{\log{M}}{R}:\log{|T_{X^{\ell}}|}> \gamma\right) \label{yek}\\
                                          &\leq \mathbb{P}\left(\log{\left|T_{X^{\frac{\log{M}}{R}}}\right|}>\gamma\right) \label{thisNewF}\\
                                          &\leq\mathbb{P}\left(-\log{p_{\hat{\theta}\left(X^{\frac{\log{M}}{R}}\right)}\left(X^{\frac{\log{M}}{R}}\right)}> \gamma+\frac{d}{2}\log{\frac{\log{M}}{R}}-C_2\right) \label{do}\\
                                          &\leq \mathbb{P}\left(-\log{p_{\theta^*}\left(X^{\frac{\log{M}}{R}}\right)}> \gamma+\frac{d}{2}\log{\frac{\log{M}}{R}}-C_2\right)\label{se}\\
                                          &=\mathbb{P}\left(\frac{-\log{p_{\theta^*}\left(X^{\frac{\log{M}}{R}}\right)}-\frac{\log M}{R}H}{\sigma\sqrt{\frac{\log M}{R}}}> \frac{\gamma+\frac{d}{2}\log{\frac{\log{M}}{R}}-C_2-\frac{\log M}{R}H}{\sigma\sqrt{\frac{\log M}{R}}}\right)\nonumber \\
                                          &\leq Q\left(\frac{\gamma+\frac{d}{2}\log{\frac{\log{M}}{R}}-C_2-\frac{\log M}{R}H}{\sigma\sqrt{\frac{\log M}{R}}}\right)+\frac{A}{\sqrt{\frac{\log M}{R}}} \label{lastLineOverFl}
\end{align}
where (\ref{yek}) is from the condition for segment $x^\ell$ to be in the dictionary in (\ref{tyCmAlg}), (\ref{thisNewF}) holds since for $x^{\ell}$ a prefix of $x^{\ell'}$, $|T_{x^{\ell}}|\leq |T_{x^{\ell'}}|$ and furthermore we assume that $\frac{\log{M}}{R}$ is an integer, (\ref{do}) is from the quantized type class size bound in Lemma \ref{TypeClSizeLem}, (\ref{se}) is from $p_{\theta^*}(x^{\ell})\leq p_{\hat{\theta}(x^{\ell})}(x^{\ell})$, and finally (\ref{lastLineOverFl}) is an application of Lemma \ref{maxLikeBerr}. In Appendix \ref{sec::solvRAschv}, we show that for the rate $R$ specified below, (\ref{lastLineOverFl}) and subsequently the overflow probability falls below $\epsilon$:
\begin{equation}
R=H+\sigma\sqrt{\frac{H}{\log{M}}}Q^{-1}(\epsilon)+H\frac{d}{2}\frac{\log\log M}{\log M}+\mathcal{O}\left(\frac{1}{\log{M}}\right).
\end{equation}
Due to the definition of $\epsilon$-coding rate, $R_M(\epsilon,\phi,p_{\theta^*})\leq R$. This completes the achievability proof.
\section{Converse}
\label{sec::Converse}
We first introduce notations relevant to the F-V length codes. Recall that any F-V length prefix code $\phi^{\texttt{FV}}$ is a mapping from a set of words $\mathcal{W}_n$, the set of all sequences of fixed input length $n$ over the alphabet $\mathcal{X}$, to variable length binary sequences. For an infinite length sequence $X^{\infty}$ emitted from the source, we adopt a one-shot setting and let
\begin{equation}
 \ell\left(\phi^{\texttt{FV}}\left(X^{\infty}\right)\right) :=\ell\left(\phi^{\texttt{FV}}\left(X_0^n\right)\right) \label{fvLengthIs}
\end{equation}
where $X_0^n\in\mathcal{W}_n$ is the prefix of $X^{\infty}$ within the set of words. For simplicity of notation, we denote $\ell^{\texttt{FV}}(X^{\infty}):= \ell(\phi^{\texttt{FV}}(X^{\infty}))$.

Let $\phi^{\texttt{VF}}$ be an arbitrary V-F length code with $M$ dictionary segments and length function $\ell^{\texttt{VF}}(\cdot)$ defined as in (\ref{OneShotELnmgthEq}). Let $R$ be any achievable $\epsilon$-coding rate for $\phi^{\texttt{VF}}$. We show that
\begin{equation}
\label{converseEq}
R\geq H+\sigma\sqrt{\frac{H}{\log{M}}}Q^{-1}(\epsilon)+H\frac{d}{2}\frac{\log\log M}{\log M}-C\frac{\log\log\log M}{\log M}.
\end{equation}

Assume $\log{M}$ and $\frac{\log{M}}{R}$ are integers. This assumption does not hurt generality of our result. It is shown in \cite{merhavVFvsFV} that for any V-F length code $\phi^{\texttt{VF}}$ with $M$ dictionary segments and length function $\ell^{\texttt{VF}}(\cdot)$, one can construct a F-V length prefix code $\phi^{\texttt{FV}}$ with $|\mathcal{X}|^\frac{\log{M}}{R}$ codewords (i.e. fixed input length of $\frac{\log{M}}{R}$) and length function $\ell^{\texttt{FV}}(\cdot)$, such that the event $\left\{\ell^{\texttt{VF}}(X^{\infty})<\frac{\log{M}}{R}\right\}$ for $\phi^{\texttt{VF}}$ is equivalent to the event $\left\{\ell^{\texttt{FV}}(X^{\infty})>\log{M}\right\}$ for $\phi^{\texttt{FV}}$. Their construction goes as follows:
\begin{itemize}
\item \textbf{Step 1:} Consider the complete $|\mathcal{X}|$-ary tree with $M$ leaves corresponding to the complete and proper V-F length code. All the dictionary segments of length greater than $\frac{\log{M}}{R}$, are shortened to $\frac{\log{M}}{R}$ letters, by pruning all subtrees with roots at depth $\frac{\log{M}}{R}$. Therefore, all the leaves (i.e. segments) of the modified tree have length at most $\frac{\log{M}}{R}$, and moreover the probability $\mathbb{P}\left(\ell^{\texttt{VF}}(X^{\infty})<\frac{\log{M}}{R}\right)$ of the modified tree is equal to that of the original tree.
\item \textbf{Step 2:} Every segment $x^*$ of the modified tree with length $\ell(x^*)<\frac{\log{M}}{R}$ is extended to $\frac{\log{M}}{R}$ by all $|\mathcal{X}|^{\frac{\log{M}}{R}-\ell(x^*)}$ possible suffixes, and accordingly, the $\log{M}$-bit codeword for this segment is also extended by all possible $\left(\left(\frac{\log{M}}{R}-\ell(x^*)\right)\left\lceil\log{|\mathcal{X}|}\right\rceil\right)$-bit suffixes. This results in a F-V length code with fixed input-length $\frac{\log{M}}{R}$ and length function $\ell^{\texttt{FV}}(\cdot)$ satisfying the required properties.
\end{itemize}

Therefore we have
\begin{align}
\mathbb{P}\left(\frac{\log M}{\ell^{\texttt{VF}}(X^{\infty})}>R\right) &= \mathbb{P}\left(\ell^{\texttt{VF}}(X^{\infty})<\frac{\log M}{R}\right) \nonumber \\
                                                    &=\mathbb{P}\left(\ell^{\texttt{FV}}(X^{\infty})>\log{M}\right) \nonumber \\
                                                    &=\mathbb{P}\left(\frac{\ell^{\texttt{FV}}(X^{\infty})}{\frac{\log{M}}{R}}>R\right). \label{lessEpsilonFVVF}
\end{align}
Since $R$ is $\epsilon$-achievable for $\phi^{\texttt{VF}}$, therefore $\mathbb{P}\left(\frac{\log M}{\ell^{\texttt{VF}}(X^{\infty})}>R\right)\leq \epsilon$ and hence (\ref{lessEpsilonFVVF}) implies
\begin{equation}
\label{ovFRelEq}
\mathbb{P}\left(\frac{\ell^{\texttt{FV}}(X^{\infty})}{\frac{\log{M}}{R}}>R\right)\leq \epsilon.
\end{equation}
Define the $\epsilon$-coding rate $R(\phi^{\texttt{FV}},\epsilon,p)$ of the F-V length code $\phi^{\texttt{FV}}$ as \cite[Eq. 9]{kosutJournal}
\begin{equation*}
R(\phi^{\texttt{FV}},\epsilon,p)=\min\left\{R_0:\mathbb{P}\left(\frac{\ell^{\texttt{FV}}(X^{\infty})}{\frac{\log{M}}{R}}>R_0\right)\leq \epsilon\right\}.
\end{equation*}
Note that the fixed input length of $\phi^{\texttt{FV}}$ is $\frac{\log{M}}{R}$. Therefore, (\ref{ovFRelEq}) implies
\begin{equation}
\label{RgreterthanRFV}
R\geq R(\phi^{\texttt{FV}},\epsilon,p).
\end{equation}
The converse for fixed-to-variable length prefix codes \cite[Theorem 15]{kosutJournal}, in turn implies\footnote{The result in \cite{kosutJournal} is stated for the class of all memoryless sources. However, adapting their proof for the exponential family is straightforward.}
\begin{equation}
\label{converseFVPRef}
R(\phi^{\texttt{FV}},\epsilon,p)\geq H+\frac{\sigma}{\sqrt{\frac{\log{M}}{R}}}Q^{-1}(\epsilon)+\frac{d}{2}\frac{\log{\frac{\log{M}}{R}}}{\frac{\log{M}}{R}}-\mathcal{O}\left(\frac{\log\log{\frac{\log M}{R}}}{\frac{\log{M}}{R}}\right).
\end{equation}
Combining (\ref{RgreterthanRFV},\ref{converseFVPRef}) yields
\begin{equation}
\label{thisEWEq}
R\geq H+\frac{\sigma}{\sqrt{\frac{\log{M}}{R}}}Q^{-1}(\epsilon)+\frac{d}{2}\frac{\log{\frac{\log{M}}{R}}}{\frac{\log{M}}{R}}-C\left(\frac{\log\log{\frac{\log M}{R}}}{\frac{\log{M}}{R}}\right)
\end{equation}
where $C$ is a constant. Through a similar iterative approach as in Appendix \ref{sec::solvRAschv}, one can show that (\ref{thisEWEq}) leads to (\ref{converseEq}).

\section{Conclusion}
\label{sec::cncld}
We derived the fundamental limits of universal variable-to-fixed length coding of $d$-dimensional exponential families of distributions in the fine asymptotic regime, where the law of large numbers may not hold. We proposed the Type Complexity code and further showed that the quantized type implementation of the Type Complexity code achieves the optimal third-order coding rate. Studying the behavior of the non-proper codes is an interesting future direction.

{\bibliographystyle{IEEEtran}
\bibliography{reputation}}

\begin{thebibliography}{10}
\providecommand{\url}[1]{#1}
\csname url@samestyle\endcsname
\providecommand{\newblock}{\relax}
\providecommand{\bibinfo}[2]{#2}
\providecommand{\BIBentrySTDinterwordspacing}{\spaceskip=0pt\relax}
\providecommand{\BIBentryALTinterwordstretchfactor}{4}
\providecommand{\BIBentryALTinterwordspacing}{\spaceskip=\fontdimen2\font plus
\BIBentryALTinterwordstretchfactor\fontdimen3\font minus
  \fontdimen4\font\relax}
\providecommand{\BIBforeignlanguage}[2]{{%
\expandafter\ifx\csname l@#1\endcsname\relax
\typeout{** WARNING: IEEEtran.bst: No hyphenation pattern has been}%
\typeout{** loaded for the language `#1'. Using the pattern for}%
\typeout{** the default language instead.}%
\else
\language=\csname l@#1\endcsname
\fi
#2}}
\providecommand{\BIBdecl}{\relax}
\BIBdecl

\bibitem{tunstall}
B.~P. Tunstall, \emph{Synthesis of noiseless compression codes}.\hskip 1em plus
  0.5em minus 0.4em\relax Ph.D. dissert., Georgia Inst. of Technol., Atlanta,
  GA, 1967.

\bibitem{drmota}
M.~Drmota, Y.~A. Reznik, and W.~Szpankowski, ``Tunstall code, khodak
  variations, and random walks,'' \emph{Information Theory, IEEE Transactions
  on}, vol.~56, no.~6, pp. 2928--2937, June 2010.

\bibitem{krichevsky}
R.~Krichevsky and V.~Trofimov, ``The performance of universal encoding,''
  \emph{Information Theory, IEEE Transactions on}, vol.~27, pp. 199--207, 1981.

\bibitem{Lawrence}
J.~Lawrence, ``A new universal coding scheme for the binary memoryless
  source,'' \emph{Information Theory, IEEE Transactions on}, vol.~23, pp.
  466--472, 1977.

\bibitem{tjalkens}
T.~Tjalkens and F.~Willems, ``A universal variable-to-fixed length source code
  based on lawrence's algorithm,'' \emph{Information Theory, IEEE Transactions
  on}, vol.~38, pp. 247--253, 1992.

\bibitem{viswer}
K.~Visweswariah, S.~R. Kulkarni, and S.~Verdu, ``Universal variable-to-fixed
  length source codes,'' \emph{Information Theory, IEEE Transactions on},
  vol.~47, pp. 1461--1472, 2001.

\bibitem{merhavVFvsFV}
N.~Merhav and D.~L. Neuhoff, ``Variable-to-fixed length codes provide better
  large deviations performance than fixed-to-variable length codes,''
  \emph{IEEE Transactions on Information Theory}, vol.~38, no.~1, pp. 135--140,
  1992.

\bibitem{nematArxiv}
N.~Iri and O.~Kosut, ``Fine asymptotics for universal one-to-one compression of
  parametric sources,'' \emph{arXiv preprint arXiv:1612.06448}, 2016.

\bibitem{nemat2}
------, ``A new type size code for universal one-to-one compression of
  parametric sources,'' in \emph{2016 IEEE International Symposium on
  Information Theory (ISIT)}, 2016, pp. 1227--1231.

\bibitem{oliver}
O.~Kosut and L.~Sankar, ``Universal fixed-to-variable source coding in the
  finite blocklength regime,'' in \emph{Information Theory Proceedings (ISIT),
  2013 IEEE International Symposium on}, 2013, pp. 649--653.

\bibitem{niriCiss18}
N.~Iri and O.~Kosut, ``Universal coding with point type classes,'' in
  \emph{51st Annual Conference on Information Sciences and Systems (CISS)},
  March 2017.

\bibitem{kosutJournal}
O.~Kosut and L.~Sankar, ``Asymptotics and non-asymptotics for universal
  fixed-to-variable source coding,'' \emph{IEEE Transactions on Information
  Theory}, vol.~63, no.~6, pp. 3757--3772, June 2017.

\bibitem{rissanenUIPE}
J.~Rissanen, ``Universal coding, information, prediction, and estimation,''
  \emph{Information Theory, IEEE Transactions on}, vol.~30, no.~4, pp.
  629--636, Jul 1984.

\bibitem{merhav}
N.~Merhav and M.~Weinberger, ``On universal simulation of information sources
  using training data,'' \emph{Information Theory, IEEE Transactions on},
  vol.~50, no.~1, pp. 5--20, Jan 2004.

\bibitem{verdulossless}
I.~Kontoyiannis and S.~Verd\'{u}, ``Optimal lossless data compression:
  Non-asymptotics and asymptotics,'' \emph{Information Theory, IEEE
  Transactions on}, vol.~60, no.~2, pp. 777--795, Feb 2014.

\bibitem{saito}
S.~Saito, N.~Miya, and T.~Matsushima, ``Evaluation of the minimum overflow
  threshold of {Bayes} codes for a {Markov} source,'' in \emph{Information
  Theory and its Applications (ISITA), 2014 International Symposium on}.\hskip
  1em plus 0.5em minus 0.4em\relax IEEE, 2014, pp. 211--215.

\end{thebibliography}

\begin{appendices}
\section{Proof of $|\mathcal{A}|\leq \ell^{d-1}$}
\label{app::SizeAi}
The type class size bounds in Lemma \ref{TypeClSizeLem} implies the following subset relationships
\begin{equation}
\left\{T\in\mathcal{T}_{\ell}:\log{|T|}\leq\gamma\right\} \subseteq \left\{T\in\mathcal{T}_{\ell}: \exists x^{\ell}\in T \mbox{ with } -\log{p_{\hat{\theta}(x^{\ell})}(x^{\ell})}+C_1\leq\gamma+\frac{d}{2}\log{\ell}\right\} \label{firstSubset}
\end{equation}
and
\begin{align}
&\left\{T\in\mathcal{T}_{\ell}:\exists\: x^{\ell}\in T \mbox{ and } x_{\ell+1}\in\mathcal{X} \mbox{ with }\log{|T_{x^{\ell}x_{\ell+1}}|>\gamma}\right\}\nonumber \\ &\hspace{0.5in}\subseteq \left\{T\in\mathcal{T}_{\ell}:\exists\: x^{\ell}\in T \mbox{ and } x_{\ell+1}\in\mathcal{X} \mbox{ with } -\log{p_{\hat{\theta}(x^{\ell}x_{\ell+1})}(x^{\ell}x_{\ell+1})}+C_2>\gamma+\frac{d}{2}\log{\ell}\right\}. \label{secondSubset}
\end{align}

Hence Lemma \ref{dataIncInfo}, along with (\ref{firstSubset},\ref{secondSubset}) and the definition of $\mathcal{A}$, imply
\begin{align*}
\mathcal{A}\subseteq \left\{T\in\mathcal{T}_{\ell}: \exists x^{\ell}\in T \mbox{ with }\gamma+\frac{d}{2}\log{\ell}-C_2-C_4<-\log{p_{\hat{\theta}(x^{\ell})}(x^{\ell})}<\gamma+\frac{d}{2}\log{\ell}-C_1\right\}.
\end{align*}
On the other hand it is shown in \cite[Eq. 32]{nematArxiv} that
\begin{equation*}
\left|\left\{T\in\mathcal{T}_{\ell}: \exists x^{\ell}\in T \mbox{ with }\gamma+\frac{d}{2}\log{\ell}-C_2-C_4<-\log{p_{\hat{\theta}(x^{\ell})}(x^{\ell})}<\gamma+\frac{d}{2}\log{\ell}-C_1\right\}\right|\leq\ell^{d-1}.
\end{equation*}
This completes the proof.

\section{Achievable $\epsilon$-coding Rate}
\label{sec::solvRAschv}
In order for (\ref{lastLineOverFl}) to be less than or equal to $\epsilon$, it must hold that
\begin{equation*}
\gamma-\frac{\log{M}}{R}H+\frac{d}{2}\log{\frac{\log{M}}{R}}-C_2\leq \sigma\sqrt{\frac{\log{M}}{R}}Q^{-1}\left(\epsilon-\frac{A}{\sqrt{\frac{\log{M}}{R}}}\right).
\end{equation*}
Recalling the designed value for $\gamma$ in (\ref{gammaValued}) along with the Taylor expansion of $Q^{-1}(\cdot)$ around $\epsilon$ yield
\begin{equation}
\label{thisljEq}
R\leq H +\sigma\sqrt{\frac{R}{\log{M}}}Q^{-1}(\epsilon)+R\frac{d}{2}\frac{\log\log{M}}{\log{M}}+\frac{C}{\log M}
\end{equation}
for some constant $C$. Define $R^*$ as the largest $R$ satisfying (\ref{thisljEq}). We then solve iteratively for $R^*$. For large enough $M$, one can show that $R^*= H+\delta_1$, where $\delta_1=o(1)$. Substituting $R^*$ in (\ref{thisljEq}) and cancelling $H$ from the left and right side of (\ref{thisljEq}), one can show that $\delta_1= \sigma\sqrt{\frac{H}{\log{M}}}Q^{-1}(\epsilon)+\delta_2$, where $\delta_2=o\left(\frac{1}{\sqrt{\log{M}}}\right)$ and the Taylor expansion of $\sqrt{H+\delta_1}$ around $H$ is employed. Finally, substituting $R^*= H+\sigma\sqrt{\frac{H}{\log{M}}}Q^{-1}(\epsilon)+\delta_2$ in (\ref{thisljEq}) and cancelling the $H$ and  $\sigma\sqrt{\frac{H}{\log{M}}}Q^{-1}(\epsilon)$ terms from the left and right side of (\ref{thisljEq}), one can show that $\delta_2= \frac{d}{2}H\frac{\log\log{M}}{\log{M}}+\delta_3$, where $\delta_3=\mathcal{O}\left(\frac{1}{\log{M}}\right)$.
\end{appendices}

\end{document}